\documentclass[aps,superscriptaddress]{revtex4}

\newcommand{\ket}[1]{|#1\rangle}

\usepackage{comment}
\usepackage{amsmath,amsthm,amsfonts,amssymb} 
\usepackage{fontenc}
\usepackage[all]{xy}
\usepackage{graphicx,subfigure,multirow}
\usepackage{tikz}
\usetikzlibrary{shapes, arrows, shadows}
\usepackage[latin1]{inputenc}
\usepackage[scaled]{helvet}
\usepackage[toc, page]{appendix}
\usepackage{graphicx}				
\usepackage{amssymb}
\usepackage{changes}
\usepackage{lipsum}
\newtheorem{definition}{Definition}
\newtheorem{theorem}{Theorem}
\newtheorem{proposition}{Proposition}

\definecolor{matty}{rgb}{0,0,1}

\input{Qcircuit.tex}

\begin{document}

	\title{Non-Unitary Quantum Computation in the Ground Space of Local Hamiltonians}
	\author{Na\"iri Usher}
	\email{ucapnus@ucl.ac.uk}
	\affiliation{Department of Physics and Astronomy, University College London, Gower Street, London WC1E 6BT, United Kingdom.}
	\author{Matty J. Hoban}
	\affiliation{University of Oxford, Department of Computer Science, Wolfson Building, Parks Road, Oxford OX1 3QD, United Kingdom.}
	\affiliation{School of Informatics, University of Edinburgh, 10 Crichton Street, Edinburgh EH8 9AB, UK}
	\author{ Dan E. Browne}
	\affiliation{Department of Physics and Astronomy, University College London, Gower Street, London WC1E 6BT, United Kingdom.}
	\begin{abstract}
A central result in the study of Quantum Hamiltonian Complexity is that the $k$-\textsc{local hamiltonian} problem is $\mathsf{QMA}$-complete \cite{KitaevBook}. In that problem, we must decide if the lowest eigenvalue of a Hamiltonian is bounded below some value, or above another, promised one of these is true. Given the ground state of the Hamiltonian, a quantum computer can determine this question, even if the ground state itself may not be efficiently quantum preparable. Kitaev's proof of $\mathsf{QMA}$-completeness encodes a unitary quantum circuit in $\mathsf{QMA}$ into the ground space of a Hamiltonian. However, we now have quantum computing models based on measurement instead of unitary evolution, furthermore we can use post-selected measurement as an additional computational tool. In this work, we generalise Kitaev's construction to allow for non-unitary evolution including post-selection. Furthermore, we consider a type of post-selection under which the construction is consistent, which we call tame post-selection. We consider the computational complexity consequences of this construction and then consider how the probability of an event upon which we are post-selecting affects the gap between the ground state energy and the energy of the first excited state of its corresponding Hamiltonian. We provide numerical evidence that the two are not immediately related, by giving a family of circuits where the probability of an event upon which we post-select is exponentially small, but the gap in the energy levels of the Hamiltonian decreases as a polynomial.
	\end{abstract}
\maketitle

\section{Introduction}
The advent of quantum information has brought the fields of theoretical computer science and physics closer together in new and exciting ways. In particular, it has been shown that key problems in condensed matter theory, such as finding the ground state of a class of Hamiltonians, can be studied through the lens of quantum computational complexity \cite{arora2009computational}.

Kitaev, building upon ideas originally developed by Feynman \cite{F85qmc}, was the first to connect the two together. He showed that determining whether a system has ground state energy $E_0$ that is less than $a$ or greater than $b$ is a hard problem for a quantum computer. However, once given a candidate ground state, a quantum computer can efficiently check its energy. More formally, Kitaev thus defined the $k$-\textsc{Local Hamiltonian} problem which he subsequently proved is $\mathsf{QMA}$-complete, the quantum computing analogue of the class $\mathsf{NP}$ \cite{KitaevBook, kempe2006complexity}. Intuitively, $\mathsf{NP}$-complete problems are hard to solve on a classical computer, even though candidate solutions --- or proofs --- can be efficiently checked, and thus similarly, $\mathsf{QMA}$-complete problems will be hard for a quantum computer to solve. Since then, building on Kitaev's seminal work, the field of Quantum Hamiltonian Complexity has flourished \cite{kempe20033, bravyi2003commutative, bravyi2006complexity, cubitt2014complexity, breuckmann2014space, gottesman2009quantum}. 

At the heart of the proof of $\mathsf{QMA}$-completeness of the $k$-\textsc{Local Hamiltonian} lies the circuit-to-Hamiltonian construction, which maps the unitary evolution of quantum states described by a quantum circuit to the ground states of a Hamiltonian operator. This procedure thus effectively encodes the computation within the Hamiltonian's ground-space, or kernel when thought of as a linear operator. The key idea is that accepting computations will have low energy eigenstates (below $a$), whereas rejecting ones will not (above $b$). The gap $a-b$ in energies is taken to be lower-bounded by an inverse polynomial in the size of the problem input. The associated $\mathsf{QMA}$-hard problem is thus to approximate the ground state energy to polynomial accuracy. 

Since Kitaev's work, the framework and models of quantum computation have evolved beyond the unitary quantum circuit model, with many of these models being motivated by ideas in physics. Such  examples include the Knill-Laflamme-Milburn (KLM) \cite{knill2001scheme} scheme for universal quantum computation with linear optics, and Measurement-based Quantum Computing (MBQC) \cite{raussendorf2001one, briegel2009measurement}, whereby universal quantum computation is achieved by a sequence of (single-qubit) measurements made on an entangled resource state. Crucially, the circuit model and MBQC can simulate one another, and are of equivalent computational power. 

The tool of postselection was introduced by Aaronson \cite{aaronson2005quantum}, who showed that given the ability to post-select, quantum circuits could solve $\mathsf{PP}$-complete problems. $\mathsf{PP}$ is a powerful complexity class, containing both $\mathsf{NP}$ and $\mathsf{QMA}$ \cite{aaronson2005quantum, marriott2005quantum, vyalyi2003qma}. Postselection similarly boosts the power of classical computation, although it is interesting to note that quantum computing taken in conjunction with postselection is more powerful than with just classical computation. 

The probability of the event upon which we post-select can be exponentially small in the input size. Indeed, otherwise the post-selected computation could be simulated by running the computation a polynomial number of times. Yet, at the same time, post-selection is a useful tool in quantum computational complexity. For example, there is now a growing body of evidence showing that sampling from the distributions produced by restricted, non-universal, quantum circuits is hard for a classical computer. By adding post-selection to our computational toolbox, sampling problems such as \textsc{BosonSampling} \cite{aaronson2011computational} and \textsc{IQPSampling} \cite{bremner2010classical, hoban2014measurement} can be shown to be hard to efficiently classically simulate. 

In the following, we investigate the construction of natural versions of the $k$-\textsc{local hamiltonian} problem from non-unitary quantum computation, by considering evolution via measurements, as seen in MBQC. Going further, we introduce a family of Hamiltonians encoding evolution via renormalised measurements, which allows us to add postselection to these non-unitary circuits. This in turn motivates our study of a restricted form of postselection, which we call tame, that limits the computational complexity of such families. Finally, we consider \textit{tame post-selection gadgets}, such as the Hadamard gadget from \textsc{IQPSampling} \cite{bremner2010classical}, and study its associated Hamiltonian. By considering two different yet computationally equivalent circuits, we numerically study the scaling of the gap $b-a$ with input size and find two radically different behaviours. In one case, the gap scales as an inverse exponential while in the other, it scales as an inverse polynomial.  This suggests that the probability of the post-selected event succeeding is not immediately connected to the gap in the Hamiltonian, which makes the connection between the $k$-\textsc{local hamiltonian} problem and non-unitary computation very subtle. It also leads to the possibility of studying post-selected quantum computation that can be encoded in physical systems that have a ``well-behaved " spectral gap. 

The paper proceeds as follows: in Sec.\ref{sec2}, we review the $k$-\textsc{local hamiltonian} problem and its associated computational complexity class $\mathsf{QMA}$. In particular, we highlight the role of the circuit-to-Hamiltonian construction in Kitaev's original proof that $k$-\textsc{local hamiltonian} is $\mathsf{QMA}$-complete. In Sec. \ref{sec3}

we introduce the formalism for constructing versions of the $k$-\textsc{local hamiltonian} problem from quantum circuits with post-selection. 

Given this construction, we then consider what we call tame post-selection as motivated by the study of MBQC in Sec. \ref{sec4} and show that the computational complexity of a version of $\mathsf{QMA}$ where the circuits include tame post-selection is only as powerful post-selected quantum computing alone. We then discuss the connection between these forms of computations and forms of the $k$-\textsc{local hamiltonian} problem with an exponentially small gap $a-b$. We then go on to give numerical evidence that certain post-selected circuits (where we post-select on exponentially unlikely events) can still give rise to Hamiltonians with gaps that are inverse polynomially bounded. 
 
 Finally, in Sec.\ref{sec5} we end with some discussion about future directions of research.

\section{The k-Local Hamiltonian Problem}\label{sec2}
The dynamics of many-body quantum systems are described by a Hamiltonian operator. Typical Hamiltonians studied within condensed matter physics are described as a sum of Hermitian operators that act on a number $k$ of subsystems, which are then said to be $k$-local. Despite this simplicity, the matrix representation of a Hamiltonian operator acting on $n$ qubits is of dimension exponential in $n$. Therefore, computing its properties, such as the ground state or ground state energy by brute force diagonalisation can be hard. Although it may seem natural to use quantum computers to compute such properties of Hamiltonians, Quantum Hamiltonian Complexity tells us that this would be still be a hard task. In order to further understand this we must formalise the problem of interest and its associated complexity class.

\begin{definition}
The \textsc{$k$-Local Hamiltonian} problem: given a $k$-local Hermitian operator $H=\sum_{i=1}^{r(n)}H_i$, $||H_i|| \leq t(n)$ and two real numbers $0 \leq a \leq b$ such that $b-a > \frac{1}{s(n)}$, and $r(n)$, $s(n)$, and $t(n)$ are polynomials, determine if $\lambda_{min}(H)<a$ or $\lambda_{min}(H) > b$, given the promise that one of these is the case and where $\lambda_{min}(H)$ denotes the smallest eigenvalue of the operator $H$. 
\end{definition}

This problem is $\mathsf{QMA}$-complete. That is, it is in the complexity class $\mathsf{QMA}$ and every problem in $\mathsf{QMA}$ may be reduced to it. For the sake of completeness, we present the definition of this complexity class. A computation begins with a classical input $x$ od size $n$ encoded in binary representation, and ends with a single bit as output. In general, we wish to determine whether the input is a $yes$-instance (output bit is $1$) or a $no$-instance (output bit is $0$). That is, whether it belongs to the set of strings $\mathcal{L}_{yes}$ whose output is accepting or to the set of rejecting outputs  $\mathcal{L}_{no}$, promised that it does indeed belong to one of these two sets.  

The class $\mathsf{QMA}$ stands for \textit{Quantum Merlin Arthur} where we imagine that Arthur, who has limited computational power wishes to determine whether a given input $x$ belongs to a language. To do so, he not only has access to a quantum computer, but also to a quantum state, the alleged proof $|\psi\rangle$. This proof state is given to him by Merlin, a computationally unbounded agent. Effectively, this is the quantum probabilistic analogue of the class $\mathsf{NP}$, whose $yes$-instances have polynomial size proofs.  

We need to be more precise about the quantum computation that Arthur does on the quantum state $|\psi\rangle$. First, there are $v(n)$ qubits in the quantum state $|\psi\rangle$, and the quantum computer is a quantum circuit with $w(n)$ quantum gates acting on $y(n)+v(n)$ qubits, where $v(n)$, $w(n)$, and $y(n)$ are polynomial functions in $n$. The description of the quantum gates and the specification of their sequence is efficiently generated by a classical computer in time at most polynomial in $n$, thus giving a so-called \textit{uniform family of quantum circuits} $\{V_{x}\}$ that have a description which depends on the input $x$. The circuits take $|\psi\rangle|00...0\rangle$ as an input quantum state, where $|00...0\rangle$ is the state of $y(n)$ qubits initialised in the state $|0\rangle$. Then after all of the $w(n)$ gates have been applied, there is a measurement in the computational basis $\{|0\rangle,|1\rangle\}$ made on the first qubit to decide the classical output of the computation: the outcome of this measurement is represented by a bit $q_{out}\in\{0,1\}$. This is typically a probabilistic process, and so we allow an error probability $\epsilon$ for Arthur's computer to output the wrong answer. We now have the ingredients for defining $\mathsf{QMA}$.

\begin{definition}
A promise problem $\mathcal{L}=(\mathcal{L}_{yes},\mathcal{L}_{no})$ is in $\mathsf{QMA}$ if for an input $x\in\{0,1\}^{n}$, there exists a uniform family of quantum circuits $\{V_{x}\}$ taking $|\psi\rangle|00...0\rangle$ as input, and with bit-value $q_{out}\in\{0,1\}$ as an outcome of a measurement on the first qubit in the basis $\{|q_{out}\rangle\}$, such that
\begin{align*}
&x  \in \mathcal{L}_{yes} \quad \textrm{if} \quad \exists |\psi\rangle, \quad \textrm{such that} \quad \text{P}[q_{out}=1] \geq \alpha , \\
&x   \in \mathcal{L}_{no} \quad \textrm{if} \quad \forall |\psi \rangle, \quad \text{P}[q_{out}=1] \leq \beta,
\end{align*}
such that $\alpha-\beta\geq 1/poly(n)$, where $poly$ is a polynomial.
\end{definition}

It should be noted that $\mathsf{NP}$ is contained in $\mathsf{QMA}$. In the next subsection we will give an overview of the proof that \textsc{$k$-Local Hamiltonian} is $\mathsf{QMA}$-complete. The central idea is to build a Hamiltonian whose ground state encodes the computation performed in a $\mathsf{QMA}$ computation such that the energy of the Hamiltonian is bounded below $a$ if and only if $x  \in \mathcal{L}_{yes}$, and above $b$ if and only if $x   \in \mathcal{L}_{no}$.

\subsection{$\mathsf{QMA}$-completeness of the \textsc{$k$-Local Hamiltonian} Problem}
The proof that the \textsc{$k$-Local Hamiltonian} problem is $\mathsf{QMA}$-complete proceeds in two parts: first, it is shown to be in $\mathsf{QMA}$ and then it is shown to be  $\mathsf{QMA}$-hard. The first part relies on effectively sampling the Hamiltonians energy given copies of the ground state. If there exist states with energy below $a$, then a quantum computer will be able to check this is correct, provided it is given an efficient description of both the ground state as a proof state. On the other hand, if no low energy state exists, then Merlin cannot send any state to convince us that the ground state energy is near-zero. Here, we shall focus on the hardness proof which relies on the circuit-to-Hamiltonian construction, whereby Arthur's quantum computation is translated into a Hamiltonian such that its lowest-eigenvalue eigenstate describes the evolution of a quantum state during that computation.

Feynman had the original insight that the discrete time evolution of a quantum system can be encoded in a Hamiltonian by constructing an operator whose kernel contains each evolution state \cite{F85qmc}. The system is assumed to be in an initial state $|\psi_0\rangle$ at time step $0$, and evolves, via a sequence of intermediate states $|\psi_i\rangle$, to a final state $|\psi_L\rangle$ at time step $L$. 

Now, in order to track the discrete time evolution of the system, a clock register is appended. For example, this could be a pointer particle moving to the left or to the right of a one-dimensional lattice. Here, a hop to a site to the right corresponds to a clock transition from time step $t$ to $t+1$. As in quantum computation operations are reversible, the particle is also allowed to move to the left. Thus $L$ time steps, require $L+1$ qubits, although we note that there exist more efficient clock constructions \cite{kempe20033}.

The \textit{history state} $|\eta \rangle$ corresponds to the equal superposition over all correct evolution states: 
\begin{equation*}
|\eta \rangle = \frac{1}{\sqrt{L+1}}\sum_{i=0}^L |\psi_i\rangle_{sys}  \otimes |i\rangle_{clock}, 
\end{equation*}
where suffices indicate the quantum system and clock register. The state $|\eta \rangle$ should be contained within the kernel of the constructed Hamiltonian, which consist of operators acting on the Hilbert spaces associated with both the system and the clock as indicated by the above history state.Thus, the ground state describes the history of the computation acting on the input state $|\psi_{0}\rangle_{sys}$.

First, the unitary $V_{x}$ enacted by Arthur in a $\mathsf{QMA}$ computation is decomposed into a polynomial sequence of single and two-qubit gates $V_{x}=U_L \ldots U_1$, which are picked from a universal gate set. Thus, we may generically consider the gates to be applied sequentially, one after the other, the gate $U_j$ being applied after $j$ time steps, resulting in the input state $|\psi\rangle|00...0\rangle$ having evolved to $U_j \ldots U_1 |\psi\rangle|00...0\rangle$. The history state of the computation can thus be expressed as:
\begin{equation*}
|\eta \rangle = \frac{1}{\sqrt{L+1}} \sum_{j=1}^L U_j \ldots U_1 |\psi\rangle|00...0\rangle \otimes |j\rangle.
\end{equation*}

The next step is to construct a Hamiltonian operator $H$ such that the history state lies within its kernel. This will be made up of three Hamiltonians: an input, a propagation and an output Hamiltonian. The role of the input Hamiltonian $H_{in}$ is to verify that the ancillary qubits Arthur has access to are correctly initialised to the state $|00...0\rangle$. This is attained by having an operator that projects onto the subspace orthogonal to $|00...0\rangle$ but projects onto the clock state being $|0\rangle$. Next, the propagation Hamiltonian ensures the correct unitary operators are applied at each time step. It is defined as the sum of the individual propagation Hamiltonian terms $H_{prop}=\sum_{j=1}^L H_j$,  where $H_j$ contains the evolution from time step $j-1$ to $j$. In this time-step, a state $|\psi \rangle $ evolves to a new state as the result of a unitary operator $U_j$ being applied, giving a component of the history state proportional to:
\begin{equation*}
|\psi \rangle \otimes |j-1 \rangle + U_j |\psi \rangle \otimes |j\rangle, 
\end{equation*}
which can easily be verified to lie within the kernel of: 
\begin{equation}\label{eq1}
H_j= \frac{1}{2}(-U_j \otimes  |j\rangle \langle j-1|  -U_j^\dag \otimes  |j-1\rangle \langle j|  \notag
 + \mathbb{I} \otimes  |j\rangle \langle j| +\mathbb{I} \otimes  |j-1\rangle \langle j-1|).
\end{equation}

Finally, at time-step $L$, the output qubit $q_{out}$ is measured in the computational basis, which in the case of an accepting computation  yields the outcome $|1\rangle$. In this case, the resulting output state resides in the nullspace of $H_{out}=|0\rangle\langle0|_{q_{out}} \otimes |L\rangle \langle L|$, which acts on the output qubit and clock system and applies identity to all other systems.

The task at hand is to now compute the smallest eigenvalue of the Hamiltonian $H=H_{in} + H_{prop}+ H_{out}$. By defining the change of basis operator $W=\sum_{j=1}^L U_j \ldots U_1 \otimes |j \rangle \langle j |$ and the state $|\phi \rangle =\frac{1}{\sqrt{L+1}}\sum_{i=1}^L |i\rangle$, the history state can be expressed as $|\eta \rangle = W \left(|\psi\rangle|00...0\rangle \otimes |\phi \rangle\right)$. Thus, we can now consider  $ |\psi\rangle|00...0\rangle \otimes |\phi \rangle $ to be in the kernel of the operator $W^\dag H W$, a simpler operator which nonetheless conserves the spectrum. 

From here, the task it to show that if $x\in\mathcal{L}_{yes}$, then the minimum eigenvalue of $H$ is less than $a$, whereas if $x\in\mathcal{L}_{no}$, then it is greater than $b$, where $b-a \geq \frac{1}{\text{poly}(n)}$, see \cite{KitaevBook}. This ends our summary of some of the key ideas used in the proof of $\mathsf{QMA}$-completeness of the \textsc{$k$-Local Hamiltonian} problem. The core idea was to see how the construction of the Hamiltonian $H$ relates directly to deciding a problem in $\mathsf{QMA}$. In the next section we will generalise this construction, in particular by building a propagation Hamiltonian that allow for non-unitary evolutions, leaving the other terms in the Hamiltonian essentially unchanged.

\section{Local Hamiltonians from Post-selected Quantum Circuits}\label{sec3}

Measurements are a key component in quantum computation, and even more so in MBQC whereby they drive the computation. As both the circuit model and MBQC are equivalent in terms of computational power, the complexity class $\mathsf{QMA}$ could equally be defined as an MBQC where part of the resource state is prepared by Merlin. However, in the proof of \textsc{$k$-Local Hamiltonian} being $\mathsf{QMA}$-complete, the Hamiltonian construction is made with respect to the quantum circuit model, with a single measurement performed at the end of the computation. Thus, we now consider how non-unitary evolution due to projective measurements can be encoded into Hamiltonians.

\subsection{Evolution via Renormalised Projection}

We now consider the process whereby a projective measurement is applied to a pure state, yielding an outcome $m$. We thus now imagine the evolution of a state  $|\psi \rangle$ at time $t$ to a new state $|\psi'\rangle= L|\psi \rangle$ at time $t+1$, where the operator $L$ is proportional to a projector, that is $L=\Pi/\sqrt{\langle\psi|\Pi|\psi\rangle}$  and where $\Pi$ is a projector.

 Previously, we considered the unitary time evolution of a system from $t$ to $t+1$ and constructed the associated history state of the system and the clock register. We now follow the same approach in order to obtain a new history state, this time corresponding to an evolution obtained via measurement: 
\begin{equation}
|\eta\rangle=|\psi \rangle \otimes |t\rangle + L|\psi \rangle \otimes |t+1\rangle. 
\end{equation}
In the next result, we now spell out a Hamiltonian $H_{t}$ in which this history state $|\eta\rangle$ lives. 

\begin{proposition}\label{prop1}
	Given a projective measurement $\{\Pi, \mathbb{I}-\Pi\}$  at time-step $t$, for the measurement outcome on state $|\psi\rangle$ corresponding to projector $\Pi$ occurring with probability $p= \langle \psi |\Pi|\psi \rangle $ independent of input state $|\psi \rangle $, then the un-normalised history state $|\eta\rangle=|\psi \rangle \otimes |t\rangle + L|\psi \rangle \otimes |t+1\rangle$ lies in the kernel of
	\begin{equation}
	H_t=\frac{p}{p+1}\Big[ L \otimes \Big(\frac{1}{\sqrt{p}} |t\rangle \langle t| -  |t\rangle \langle t+1|-|t+1\rangle \langle t| 
	+ \sqrt{p} |t+1\rangle \langle t+1|  \Big)\Big]+\left(\mathbb{I}-\Pi\right) \otimes |t+1\rangle \langle t+1|,
	\end{equation}
	for $L=p^{-\frac{1}{2}}\Pi$.
\end{proposition}
\begin{proof}
Since the span of the clock states $\{|t\rangle,|t+1\rangle\}$ is a two dimensional Hilbert space, the operator $H_{t}$ can be decomposed as 
\begin{equation}
H_t=N\big( H_{11} \otimes |t\rangle \langle t| +H_{12} \otimes |t\rangle \langle t+1|  
+H_{21} \otimes |t+1\rangle \langle t| +H_{22} \otimes |t+1\rangle \langle t+1| \big),
\end{equation}
where $N$ is a normalisation constant and $H_{ij}$ is an operator acting on the system with initial state $|\psi\rangle$. The requirement for the history state  $|\eta \rangle$ to satisfy $H|\eta \rangle =0$ leads to
\begin{equation}
N(H_{11} |\psi \rangle +H_{12}L|\psi \rangle \otimes |t \rangle  \notag+ H_{21} |\psi \rangle \otimes |t +1 \rangle   
 + H_{22}L|\psi \rangle \otimes |t +1 \rangle ) =0.
\end{equation}
As the operator $H_{t}$ is constrained to be Hermitian, we have that $H_{21}=H_{12}^\dag$. This produces the following system of equations: 
\begin{equation*}
\begin{cases}
H_{11} |\psi \rangle +H_{12}L|\psi \rangle =0, \\
H_{12}^\dag |\psi \rangle + H_{22}L|\psi \rangle  =0.
\end{cases}
\end{equation*}
A natural solution to the above set of equations is given by:  $H_{11}=\frac{1}{p} \Pi$, $H_{12}=-\frac{1}{\sqrt{p}}\Pi$, and $H_{22}=\Pi$, thus yielding the operator
	\begin{equation}
	H_t=N\Pi \otimes \Big(\frac{1}{p} |t\rangle \langle t| - \frac{1}{\sqrt{p}} |t\rangle \langle t+1| 
	-\frac{1}{\sqrt{p}}|t+1\rangle \langle t| +  |t+1\rangle \langle t+1|  \Big).
	\end{equation}
We now, without loss of generality, constrain the operator $H_{t}$ to be a projector, i.e. $H_{t}^2=H_{t}$. But, we note that we now have that $H^2= cH$, where:
	\begin{equation}
	H_t^2 = N^2 \Big(1 +\frac{1}{p}\Big) \Pi \otimes \Big(\frac{1}{p} |t\rangle \langle t| - \frac{1}{\sqrt{p}} |t\rangle \langle t+1| 
	-\frac{1}{\sqrt{p}}|t+1\rangle \langle t| +  |t+1\rangle \langle t+1|  \Big), 
	\end{equation}
and where $c=N(1 + p^{-1})$. This thus leads to a normalisation constant $N$ which depends on the outcome probability $p$, i.e. $N(p)=p/(p+1)$. 

By construction, the Hamiltonian $H_{t}$ contains the history state $|\eta \rangle$ in its kernel. But, the question is now whether the kernel contains any other states. And indeed, due to the orthogonality of the projectors in the measurement, there will be states lying in an orthogonal sub-space of the projector $\Pi$, that is of the form $(\mathbb{I}-\Pi) |\psi \rangle \otimes |t'\rangle$, in particular for $t'=t$ and $t'=t+1$. In order to exclude these states from the kernel, we add the following term to the Hamiltonian $H_t$ $(\mathbb{I}-\Pi)^\perp \otimes |t+1\rangle \langle t+1|$ to $H_{t}$, and we thus finally obtain: 
	\begin{equation}
	H_t= N(p)\Big[ \frac{1}{\sqrt{p}}\Pi \otimes \Big(\frac{1}{\sqrt{p}} |t\rangle \langle t| -  |t\rangle \langle t+1| -|t+1\rangle \langle t|  
	+ \sqrt{p} |t+1\rangle \langle t+1|  \Big)\Big] +(\mathbb{I}-\Pi) \otimes |t+1\rangle \langle t+1|,
	\end{equation}
for $N(p)=p/(p+1)$. This concludes the proof.
\end{proof}

We can now consider a circuit of $T$ time-steps where at each time-step either a unitary or a renormalised projector is sequentially applied to the input state $|\psi\rangle|00...0\rangle$. That is, the computation evolves in layers of unitary evolution and measurements. For example, after $T$ time-steps of a circuit in which a unitary $U_{i}$ is alternated with a renormalised projector $L_{j}$, the state of the system would be $|\psi_{T}\rangle=U_{T}L_{T-1}...L_{2}U_{1}|\psi\rangle$. The history state $|\eta\rangle$ is then:
\begin{equation*}
|\eta\rangle=\frac{1}{T+1}\Big(|\psi\rangle\otimes|0\rangle+...+U_{T}L_{T-1}...L_{2}U_{1}|\psi\rangle\otimes|T\rangle\Big),
\end{equation*}
which is in the kernel of the Hamiltonian $H=H_{in}+H_{prop}+H_{out}$ with $H_{in}$ and $H_{out}$ as defined before and now
\begin{equation}
H_{prop}=\sum_{i=0}^{T/2}H_{2i}^{unitary}+\sum_{j=1}^{T/2}H_{2j-1}^{post},
\end{equation}
with $H_{i}^{unitary}$ and $H_{j}^{post}$ being operators of the form in Eq.\eqref{eq1} and in Prop.\ref{prop1} respectively. Therefore, evolution involving measurements can be encoded into a Hamiltonian in a natural extension of the Kitaev-Feynman approach.

Of course the above evolution corresponds to conditioning on a particular outcome occurring, i.e. post-selection. Here, post-selection is modelled by the renormalised projector, but crucially relies on the quantum state to dictate the norm of this renormalised projector. Indeed this is one of the major modifications to the $k$-\textsc{Local Hamiltonian} problem when we consider post-selection. Here, the operator norm of the individual evolution terms $H_{j}^{post}$ in the Hamiltonian may not be bounded by a polynomial in the input size to the problem. For indeed, if the probability $p$ of a particular event happening is exponentially small then the operator norm will be upper-bounded by an exponential. 

\subsection{Tame Post-selection}

Clearly given evolution involving general post-selection, to construct the above Hamiltonian with the history state in its kernel we will need to know the initial quantum state $|\psi\rangle$ to calculate the renormalised projectors. However, a crucial aspect of the $k$-\textsc{Local Hamiltonian} problem is that it is defined independently of its ground state. Therefore, to get around these issues we study the concept of \textit{tame post-selection}. Here, the initial state $|\psi\rangle|00...0\rangle$ evolves to $U|\psi\rangle|00...0\rangle$ via unitary evolution, before a projective measurement $\{\Pi,\mathbb{I}-\Pi\}$ is applied to the system. We then post-select on the outcome associated with $\Pi$ occurring, and have that the probability of obtaining this outcome is independent of the initial state $|\psi\rangle$. Note that the probability of obtaining the outcome could be exponentially small in the size of $|\psi\rangle$. We emphasize that we only demand that the probability be independent of only $|\psi\rangle$; it could vary if we replace the state $|00...0\rangle$ with another (known) quantum state. 

\begin{definition}\label{tamepost}
Given a bipartite Hilbert space $\mathcal{H}=\mathcal{H}_{sys}\otimes\mathcal{H}_{anc}$ consisting of a system with space $\mathcal{H}_{sys}$, and an ancillary system with space $\mathcal{H}_{anc}$ (both with the same dimension), in an initial quantum state $\vert\phi\rangle=\vert\psi\rangle\vert0\rangle$ such that $\vert\psi\rangle\in\mathcal{H}_{sys}$ and $\vert 0\rangle\in\mathcal{H}_{anc}$, if a unitary $U$ is applied to $\vert\phi\rangle$ followed by a projective measurement $\{\Pi_{k}:=\vert k\rangle\langle k\vert\}_{k}$ with outcomes $\{k\}$ applied to the system $\mathcal{H}_{sys}$, then post-selection on a particular outcome $k'$ is \textbf{tame post-selection} if $p(k'):=\langle\psi|\langle 0|U^{\dagger}(\Pi_{k'}\otimes\mathbb{I}_{anc})U|\psi\rangle|0\rangle$ is the same for all $|\psi\rangle\in\mathcal{H}_{sys}$.
\end{definition}

An example of tame post-selection would be the \textit{Hadamard gadget}, which was used in Ref. \cite{bremner2010classical} to show the classical hardness of \textsc{IQPSampling}. This is a method of implementing a Hadamard gate via measurement and post-selection, as illustrated in Fig. $3$. Here, a qubit in an arbitrary state $|\psi \rangle$ is entangled with an ancilla (initialised in the fixed $|+\rangle$ state) via a controlled-$Z$ operator $|0\rangle\langle 0|\otimes\mathbb{I}+|1\rangle\langle 1|\otimes Z$, with $Z$ being the Pauli-$Z$ operator. Then, the first qubit is measured in the Pauli-$X$ basis $\{|+\rangle,|-\rangle\}$ and we post-select upon obtaining the outcome associated with state $|+\rangle$. This results in the state on the second qubit being $H|\psi \rangle$. Here, the probability of obtaining the measurement outcome is  $1/2$ for both outcomes, and thus is independent of the state $|\psi\rangle$. On the other hand, if we were to alter the ancilla to have another state other than $|+\rangle$, this probability could change. This helps us emphasize that tame post-selection is tame with respect to a particular `input' sub-system.
\begin{figure}[h]
	\centering
	\mbox{
		\Qcircuit @C=2em @R=2em {
			\lstick{|\psi \rangle }	& \ctrl{1} & \measureD{X} & \cw & \rstick{\ket{+}}\\
			\lstick{|+\rangle }	     & \ctrl{-1}& \qw & \qw & \rstick{H|\psi\rangle }
		} }
		\caption{\emph{Hadamard Gadget} A measurement of the Pauli-$X$ observable is performed on the first qubit. By post-selecting on obtaining the outcome $+1$ corresponding to eigenvector $|+\rangle$, a Hadamard gate is applied to the unmeasured qubit.} 
		\label{Hadamard Gadget.}
	\end{figure}
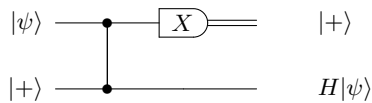
	
This post-selection results in a unitary operator being applied to the unmeasured system, a concept which is at the core of MBQC wherein unitary evolution is simulated by measurements. We generalise this kind of post-selection ``gadget'' in the following result.

\begin{proposition}
	\label{lemprob}
	Let $|\psi \rangle |E\rangle$ be a quantum state in a Hilbert space $\mathcal{H}=\mathcal{H}_{sys} \otimes \mathcal{H}_{env}$, where $|\psi \rangle\in\mathcal{H}_{sys}$ and $|E\rangle\in\mathcal{H}_{env}$, and the two Hilbert spaces have the same dimension. Suppose a unitary operator $U$ is applied to the joint system, followed by a projective measurement on $\mathcal{H}_{sys}$ in the orthonormal basis $\{ |e_k\rangle \}$, as illustrated in Fig. \ref{fig2}.  Let $p_m=\langle E|\langle\psi|\left(|e_m\rangle\langle e_{m}|\otimes\mathbb{I}\right)|\psi \rangle|E\rangle$ denote the probability of outcome $m$ occurring. Then, if $p_m$ is independent of the state $|\psi \rangle\in\mathcal{H}_{sys}$, the action of this process on the system is equivalent to applying $\sqrt{p_m} V_m$ to $|\psi\rangle$, where $V_m$ is a unitary operator. 
\end{proposition}
		\begin{proof}
The system is prepared in the state $|\psi \rangle \in \mathcal{H}_{sys}$ and is initially uncorrelated with the environment $|E \rangle \in \mathcal{H}_{env}$.  
				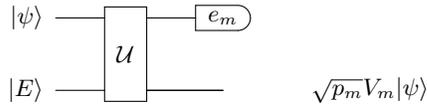
\begin{figure}[h!]
					\mbox{
							\Qcircuit @C=2em @R=2em {
								\lstick{\ket{\psi}}	& \multigate{1}{\mathcal{U}}	 &  \measureD{e_m} &  \\
								\lstick{\ket{E}}   	 &\ghost{\mathcal{U}}                & \qw  \qw &   \rstick{\sqrt{p_m}V_m |\psi \rangle }
							} }
						\caption{A unitary operator is applied to an unknown state $|\psi \rangle$ and an ancilla, before the system is measured. This results in sub-normalised unitary being applied to the ancilla.}\label{fig2}
					\end{figure}
A unitary operator $U$ acting on the joint state space  $ \mathcal{H}=\mathcal{H}_{sys} \otimes \mathcal{H}_{env}$ is applied, followed by a projective measurement of an orthonormal basis  $\{|e_k\rangle \}$ of the system, as shown in Fig. $1$. The resulting evolution (up to normalisation factors) is given by: 
			\begin{equation}
			\rho \otimes |e_0 \rangle \langle e_0| \to \Pi_{m}U\left(\rho \otimes |e_0 \rangle \langle e_0|\right)U^\dag \Pi_{m}, 
			\end{equation} 
			where $m$ denotes the obtained measurement outcome associated with projector $\Pi_{m}=|e_m\rangle \langle e_m|\otimes\mathbb{I}$, and $\rho=|\psi \rangle \langle \psi |$. We can consider the state $\rho'_{env}$ of the environment system with Hilbert space $\mathcal{H}_{env}$ after the application of $U$ and measurement $\Pi_m$, which is then the effective image of a map $\mathcal{K}_{m}$ on input state $\rho$, i.e. $\rho'=\mathcal{K}_{m}(\rho)=\textrm{tr}_{sys}\left(K_m |\psi \rangle \langle \psi| K_m^\dag\right)$ where the $K_m$ associated with outcome $m$ is defined as
\begin{equation*}			
			K_m=\Big((\Pi_{m}\otimes\mathbb{I})U\left(\mathbb{I} \otimes |E\rangle \langle E|\right)\Big).
\end{equation*}		
Therefore, upon tracing out the system in $\mathcal{H}_{sys}$, we have an effective map on the state $\rho$ of $L_m \rho L_m^\dag$ with $L_m=\langle e_{m}|U(\mathbb{I}\otimes|E\rangle \langle E|)$. Here, we demand that the outcome probability $p_{m}$ to be independent of the input state, i.e. $p_m= \langle \psi | L_m^\dag L_m | \psi \rangle$,  $\forall\textrm{ } |\psi \rangle$. The only way this is satisfiable for all input states is if $L_m^\dag L_m=p_m \mathbb{I}$. This, in turn, means that $L_m$ must be proportional to a unitary operator $V_m$, such that $L_m =\sqrt{p_m} V_m $. 
		\end{proof}

Clearly the Hadamard gadget is one such example of a process as outlined in this result. It shows that the tame post-selection we consider corresponds to applying a unitary in the input state $|\psi\rangle$ after renormalising by the probability of getting that outcome.  Thus if we wish to encode a tame post-selection into a Hamiltonian as previously outlined, then a renormalised unitary evolution will be encoded, which is a subtle alteration of the Kitaev-Feynman construction. However, as we will discuss next, this post-selection can also be used to solve very powerful computations.

\subsection{Computational Complexity of Post-selected Quantum Circuits}

If we consider quantum circuits in which post-selection is given ``for free", i.e. we can decide the property of an input conditioned on the outcome of some measurement. Aaronson was the first to define the complexity class $\mathsf{PostBQP}$ as the class of decision problems which can be decided by a quantum circuit that is of size polynomial in the input size that utilises post-selection \cite{aaronson2005quantum}. To be more formal, for an input $x$ of size $n$, a classical machine generates a description of a quantum circuit $C_{x}$ in time at most polynomial in $n$, hence it is a \textit{uniform circuit}. This quantum circuit takes as input the state $|00...0\rangle$, and has a set of post-selection qubits, and an output qubit. A measurement in the computational basis $\{\ket{0},\ket{1}\}$ is made on both the post-selection and output qubits, with the classical bit-strings $q_{post}$ and $q_{out}$ as the outcomes respectively. The circuit then post-selects on getting $q_{post}=(0,0...0):=\textbf{0}$, and then conditioned on these outcomes, the circuit decides whether to accept an input (if $q_{out}=1$) or not. We allow this decision process to fail with some non-zero probability, thus giving us the following complexity class.
\begin{definition}
	A promise problem $\mathcal{L}=(\mathcal{L}_{yes},\mathcal{L}_{no})$ is in \emph{$\mathsf{PostBQP}$} if for an input $x\in\{0,1\}^{n}$, there exists a uniform quantum circuit family $\{C_{x}\}$ with each $C_{x}$ taking $|00...0\rangle$ as input, and with post-selection and output qubits, which are all measured in the computational basis and giving outcomes as bit-strings $q_{post}$ and $q_{out}$ such that:
	\begin{align*}
	\text{if } x \in \mathcal{L}_{yes},  \quad & \text{P}[q_{out}=1|q_{post}=\bold{0}] \geq 2/3, \\
	\text{if } x \in \mathcal{L}_{no},  \quad & \text{P}[q_{out}=1|q_{post}=\bold{0}] \leq 1/3,
	\end{align*}
	where $\text{P}[q_{post}=\bold{0}] \geq 2^{-poly(n)}$ and $poly$ is some polynomial function.
\end{definition} 

We need the constraint that $\text{P}[q_{post}=\bold{0}] \geq 2^{-poly(n)}$ so that we have a consistent definition of $\mathsf{PostBQP}$: for any choice of universal gate set for the circuit $C_{x}$ we get exactly the same complexity class. This subtlety has been discussed and addressed by Kuperberg \cite{Kuperberg}. It should also be noted that the number of qubits being post-selected does not make any difference to the class as long as it is polynomial in the size of the input: for example, a circuit with only one post-selection qubit can be simulated by a circuit with $m>1$ post-selection qubits by having $m-1$ of the post-selection qubits be prepared in the state $\vert 0\rangle$ and not have any unitary act on them. Also, any circuit containing intermediate post-selected measurements can be simulated by one with only post-selection at the end on the registry by using the technique of deferring measurements at the cost of introducing a new ancilla \cite{aaronson2005quantum}. Aaronson first proved that $\mathsf{PP}\subseteq\mathsf{PostBQP}$, and claimed to prove that $\mathsf{PostBQP}\subseteq\mathsf{PP}$, but without putting a bound on the probability $\text{P}[q_{post}=\bold{0}]$. Here $\mathsf{PP}$ is the set of languages that are decided by probabilistic Turing machine with unbounded error, i.e. the error can be arbitrarily close to $1/2$. As mentioned earlier this is a relatively large complexity class in that it contains both $\mathsf{NP}$ and $\mathsf{QMA}$ (and thus $\mathsf{BQP}$).  As mentioned, Kuperberg pointed out this oversight, and if one bounds the probability as we have done, the containment $\mathsf{PostBQP}\subseteq\mathsf{PP}$ is indeed true, as pointed out in the following theorem.
\begin{theorem}
[Aaronson-Kuperberg, \cite{aaronson2005quantum,Kuperberg}] $\mathsf{PostBQP}=\mathsf{PP}$
\end{theorem}

Therefore, if we consider quantum computations with post-selection we have access to great computational power. Equally, if we encode these computations in the ground state of a Hamiltonian through the Kitaev-Feynman construction, then given access to this state we have access to this computational power. In this section we want to formalise the computational complexity of the computations that are being translated into Hamiltonians through this construction. In particular, if this computation involves post-selection, what is the computational complexity of a circuit that takes an arbitrary state (a proof state) with some fixed ancilla qubits, and subjects it to unitary evolution and post-selection as well as a final measurement?

Just as $\mathsf{QMA}$ contains $\mathsf{BQP}$, but is distinct since it allows for access to quantum proof states, we now consider the analogue of $\mathsf{QMA}$ that allows for post-selection in Arthur's computation and thus contains $\mathsf{PostBQP}$; naturally, we call this analogue $\mathsf{PostQMA}$, where we have that $\mathsf{PostBQP}\subseteq\mathsf{PostQMA}$. The main difference between $\mathsf{PostBQP}$ and $\mathsf{PostQMA}$ is that, in the latter, a proof state $|\psi\rangle$ is an input into a quantum circuit, and without loss of generality we will fix it such that the size of the post-selection register is the same size as the state $|\psi\rangle$. We will argue that we can do this after presenting the definition as follows:

\begin{definition}
	A promise problem $\mathcal{L}=(\mathcal{L}_{yes},\mathcal{L}_{no})$ is in \emph{$\mathsf{PostQMA}$} if for an input $x\in\{0,1\}^{n}$, there exists a uniform quantum circuit family $\{V_{x}\}$ with each $V_{x}$ taking $|\psi\rangle|00...0\rangle$ as input and $|\psi\rangle$ consisting of a number of qubits $w$ at most polynomial in $n$, and with post-selection and output qubits, which are all measured in the computational basis and giving outcomes as bit-strings $q_{post}\in\{0,1\}^{w}$ and $q_{out}\in\{0,1\}$ such that:
	\begin{align*}
	\text{if } x \in \mathcal{L}_{yes},  \quad & \exists |\psi\rangle \textrm{ } \text{P}[q_{out}=1|q_{post}=\bold{0}] \geq 2/3, \\
	\text{if } x \in \mathcal{L}_{no},  \quad & \forall |\psi\rangle \textrm{ } \text{P}[q_{out}=1|q_{post}=\bold{0}] \leq 1/3,
	\end{align*}
	where $\text{P}[q_{post}=\bold{0}] \geq 2^{-poly(n)}$, and $poly$ is some polynomial function.
\end{definition} 

The first thing to observe about this class is that, as expected, $\mathsf{PostBQP}$ is contained in $\mathsf{PostQMA}$, since the circuit could just `ignore' the input $|\psi\rangle$ and replace it with the all-zeroes input $|00...0\rangle$, and we recover those computations in $\mathsf{PostBQP}$. Also, as mentioned, it is not a restriction to have the post-selection register have $w$ qubits (i.e. the same size at the proof state $|\psi\rangle$). Given a circuit where post-selection is on fewer than $w$ qubits, we can pad out the size of the register with ancillas prepared in the state $|0\rangle$. If a circuit has more than $w$ qubits in the post-selection register, then we can just take the NOR of measurement outcomes to reduce to the register being of size $w$. Also, again intermediate post-selected measurements do not affect the computational complexity for the same reasons they do not alter $\mathsf{PostBQP}$. Finally, while $\mathsf{PP}=\mathsf{PostBQP}\subseteq\mathsf{PostQMA}$ it is not clear what the best upper bound on $\mathsf{PostQMA}$ is. We can give the upper bound of $\mathsf{PostQMA}\subseteq\mathsf{NEXP}$, which is the class of problems decided by a non-deterministic exponential time Turing machine. In order to simulate the $\mathsf{PostQMA}$ computation in $\mathsf{NEXP}$, a non-deterministic Turing machine guesses a classical description of the state $|\psi\rangle$, and then simulates the quantum circuit on this state in exponential time to decide whether to accept the input. We leave it open whether a better bound can be found.

In our work we are concerned with tame post-selection as defined earlier in Def. \ref{tamepost}. In the definition, we can map every element to a post-selected quantum circuit. The ancillary system $\mathcal{H}_{anc}$ is the input to the quantum circuit initiated to the state $|00...0\rangle$, the system $\mathcal{H}_{sys}$ is associated with the proof given from Merlin in state $|\psi\rangle$, the unitary $U$ acting on $\mathcal{H}_{sys}\otimes\mathcal{H}_{anc}$ is the unitary in the quantum circuit, and the post-selection register is the set of qubits of $\mathcal{H}_{sys}$ with the projector $\Pi_{k}$ being $|00...0\rangle\langle 00...0|$. Therefore, tame post-selection is the condition that the probability of $q_{post}=\bold{0}$ is the same for all states $|\psi\rangle$.

Given all of the above elements we can consider the complexity class $\mathsf{PostQMA}^{*}$ associated with post-selected circuits such that they satisfy tame post-selection. We also see the convenience of having the post-selection register being the size of the proof state since the first $w$ qubits of the circuit can be the proof system, and without loss of generality, the post-selection register is again the first $w$ qubits of the circuit. This circuit is then an example of tame post-selection as in Def. \ref{tamepost}. The following definition of $\mathsf{PostQMA}^{*}$ can now be presented.

\begin{definition}
	A promise problem $\mathcal{L}=(\mathcal{L}_{yes},\mathcal{L}_{no})$ is in \emph{$\mathsf{PostQMA^*}$} if for an input $x\in\{0,1\}^{n}$, there exists a uniform quantum circuit family $\{V_{x}\}$ with each $V_{x}$ taking $|\psi\rangle|00...0\rangle$ as input and $|\psi\rangle$ consisting of a number of qubits $w$ at most polynomial in $n$, and with post-selection and output qubits, which are all measured in the computational basis and giving outcomes as bit-strings $q_{post}\in\{0,1\}^{w}$ and $q_{out}\in\{0,1\}$ such that:
	\begin{align*}
	\text{if } x \in \mathcal{L}_{yes},  \quad & \exists |\psi\rangle \textrm{ } \text{P}[q_{out}=1|q_{post}=\bold{0}] \geq 2/3, \\
	\text{if } x \in \mathcal{L}_{no},  \quad & \forall |\psi\rangle \textrm{ } \text{P}[q_{out}=1|q_{post}=\bold{0}] \leq 1/3,
	\end{align*}
	where $\text{P}[q_{post}=\bold{0}] \geq 2^{-poly(n)}$ and is the same for all $|\psi\rangle$, and $poly$ is some polynomial function.
\end{definition} 

At first sight, given Proposition \ref{lemprob}, it might not seem obvious that $\mathsf{PostQMA}^{*}$ is a class more powerful than $\mathsf{QMA}$ since tame post-selection results in unitary evolution of $|\psi\rangle$ up to renormalization. However, following similar arguments as outlined earlier with regard to $\subseteq\mathsf{PostQMA}$, we have that $\mathsf{PostBQP}\subseteq\mathsf{PostQMA}^{*}\subseteq\mathsf{PostQMA}$, therefore $\mathsf{PP}\subseteq\mathsf{PostQMA}^{*}$. In particular, to show that $\mathsf{PostBQP}\subseteq\mathsf{PostQMA}^{*}$, Arthur can just `replace' the state $|\psi\rangle$ with the all-zeroes state, $|00...0\rangle$ and this is permitted by the definition of $\mathsf{PostQMA}^{*}$. That is, take a circuit that is used to decide a problem in $\mathsf{PostBQP}$, which consists of preparing the input state $|00..0\rangle$ of $r$ qubits, feeding it into a quantum circuit with unitary $U$, and then post-selecting on the first qubit giving outcome $0$ for a computational basis measurement, and accepting if the second qubit gives $1$ for a computational basis measurement. This can readily be turned into a $\mathsf{PostQMA}^{*}$ algorithm for an arbitrary proof state $|\psi\rangle$ of, say, $w=r$ qubits provided by Merlin (if $w\neq r$ then either the input or proof state can be padded with extra ancillas by Arthur). First, Arthur prepares $w-1$ qubits in state $|00...0\rangle:=|0_{w-1}\rangle$ and the $r$ qubits in the state $|00...0\rangle:=|0_{r}\rangle$ from the $\mathsf{PostBQP}$ computation. Arthur's initial quantum state is then $|\psi\rangle|0_{w-1}\rangle|0_{r}\rangle$, he then applies $U$ to $|0_{r}\rangle$ and identity operators to all the other qubits. After the application of these unitaries, he will measure the state the $w-1$ qubits in the state $|0_{w-1}\rangle$ and the first qubit of the system in state $U|0_{r}\rangle$ in the computational basis, and post-select on getting the outcome $\bold{0}$. This probability is independent of the state $|\psi\rangle$ since there is no interaction between Arthur's qubits prepared in the state $|0_{w-1}\rangle|0_{r}\rangle$ and $|\psi\rangle$. Finally Arthur uses a measurement of the second qubit of the state $U|0_{r}\rangle$ to accept or reject. This construction also is compatible with the definition of $\mathsf{PostQMA}^{*}$ since the post-selection register was of size $w$. 

So far the only complexity theoretic upper bound on $\mathsf{PostQMA}^{*}$ is $\mathsf{PostQMA}\subseteq\mathsf{NEXP}$, but can we do better? The next result gives a strong bound on the class $\mathsf{PostQMA}^{*}$.

\begin{theorem}
$\mathsf{PostQMA}^{*}=\mathsf{PostBQP}=\mathsf{PP}$
\end{theorem}

\begin{proof}
As discussed earlier, we have the inclusion that $\mathsf{PP}\subseteq\mathsf{PostQMA}^{*}$, so it remains to prove that $\mathsf{PostQMA}^{*}\subseteq\mathsf{PP}$. We prove this using $\mathsf{GapP}$ functions, which is the difference between the number of accepting and rejecting paths of a non-deterministic Turing machine. More formally, given a non-deterministic Turing machine $\mathcal{N}$ and input $x$, then $N_{acc}(x)$ and $N_{rej}(x)$ is the number of accepting and rejecting paths of $\mathcal{N}$ respectively given $x$, then a $\mathsf{GapP}$ function is $f(x)=N_{acc}(x)-N_{rej}(x)$ \cite{fenner}. The complexity class $\mathsf{PP}$ is defined as those languages $\mathcal{L}$ where $f$ and $g$ are $\mathsf{GapP}$ functions such that if input $x$ is in $\mathcal{L}$ then $2/3\leq f(x)/g(x)\leq 1$, and if $x\notin\mathcal{L}$ then $0\leq f(x)/g(x)\leq 1/3$.

Returning to the circuits in the class $\mathsf{PostQMA}^{*}$, the probability of accepting conditioned on a particular outcome happening is $\text{P}[q_{out}=1|q_{post}=\bold{0}]=\text{P}[q_{out}=1,q_{post}=\bold{0}]/\text{P}[q_{post}=\bold{0}]$. We will show that this conditional probability is the quotient of two $\mathsf{GapP}$ functions and thus $\mathsf{PostQMA}^{*}\subseteq\mathsf{PP}$. Without loss of generality we will take the universal gate set of the quantum circuits to be the Hadamard and Toffoli gates.

Firstly, since the definition of tame post-selection means that $\text{P}[q_{post}=\bold{0}]$ is the same for all possible states $|\psi\rangle$ then we can evaluate this probability for the case where $|\psi\rangle=|00...0\rangle$, the all-zeroes input. That is, we have a quantum circuit and wish to calculate $\text{P}[q_{post}=\bold{0}]$ for this circuit. Using a result of Fortnow and Rogers, this probability will be equal to $g(x)/2^{h(x)}$ for some $g$ being a $\mathsf{GapP}$ function and $h(x)$ being an efficiently computable function dependent on the number of Hadamard gates in the circuit for input $x$ \cite{fortnow}. However, $\text{P}[q_{out}=1,q_{post}=\bold{0}]$ may not be the same for $|\psi\rangle$ as it is for the all-zeroes input, so we need to address this issue separately.

The complicating factor for evaluating $\text{P}[q_{out}=1,q_{post}=\bold{0}]$ is that it is a probability for an input state $|\psi\rangle$, so we can calculate the maximal value of this probability (for all states $|\psi\rangle$) to decide whether the input is accepted. Firstly, we divide the qubits into the proof qubits and ancillae qubits, denoted $sys$ and $anc$ respectively. Following Vyalyi \cite{vyalyi2003qma}, we can see this maximal probability as the largest eigenvalue $\lambda_{max}$ of the operator 
\begin{equation*}
A=\textrm{tr}_{anc}\bigg(\big(|00...0\rangle\langle 00...0|_{q_{post}}\otimes|1\rangle\langle 1|_{q_{out}}\big)V_{x}\big(\mathbb{I}_{sys}\otimes|00..0\rangle\langle 00...0|_{anc}\big)V_{x}^{\dagger} \bigg).
\end{equation*}
Since $\lambda_{max}^{d}\leq\textrm{tr}(A^{d})\leq 2^{w}\lambda_{max}^{d}$ since the operator $A$ is a $2^{w}$-by-$2^{w}$ operator, if we choose $d=w+1$, then we have the following useful relationships:
\begin{eqnarray*}
\textrm{if max }\text{P}[q_{out}=1,q_{post}=\bold{0}]=(\lambda_{max}2^{h(x)})/g(x)\leq 1/3&\implies &(3/2)^{d-1}\big(2^{h(x)})\textrm{tr}(A^{d})\big)/g(x)\leq 1/3,\\
\textrm{if max }\text{P}[q_{out}=1,q_{post}=\bold{0}]=(\lambda_{max}2^{h(x)})/g(x)\geq 2/3&\implies &(3/2)^{d-1}\big(2^{h(x)})\textrm{tr}(A^{d})\big)/g(x)\geq 2/3.
\end{eqnarray*}
It was proven by Vyalyi that $\textrm{tr}(A^{d})=k(x)2^{-dh'(x)}$ where $k$ is a $\mathsf{GapP}$ function, and $h'(x)$ is another efficiently computable function dependent on the number of Hadamards in the circuit \cite{vyalyi2003qma}. Since $(3/2)^{d-1}2^{h(x)-dh'(x)}$ is an efficiently computable function and an efficiently computable function multiplied by a $\mathsf{GapP}$ function is another $\mathsf{GapP}$ function \cite{fenner}, we can define $f(x)=(3/2)^{d-1}2^{h(x)-dh'(x)}k(x)$ to obtain
\begin{eqnarray*}
x\in\mathcal{L}_{yes}& \textrm{if } f(x)/g(x)\geq 2/3,\\
x\in\mathcal{L}_{no}& \textrm{if } f(x)/g(x)\leq 1/3,
\end{eqnarray*}
if the promise problem $\mathcal{L}=(\mathcal{L}_{yes},\mathcal{L}_{no})$ is in $\mathsf{PostQMA}^{*}$, and thus $\mathsf{PostQMA}^{*}\subseteq\mathsf{PP}$.
\end{proof}

Therefore, tame post-selection restricts the computational power at hand to just be that which is found in standard post-selected quantum circuits without proof states. As a result, if we use the Kitaev-Feynman construction to build a Hamiltonian encoding a tame post-selected quantum circuit into its nullspace, a state in the nullspace just encodes problems in $\mathsf{PP}$, and nothing more powerful. This might sound very powerful since $\mathsf{QMA}\subseteq\mathsf{PP}$, with equality thought unlikely to hold \cite{vyalyi2003qma}. On the other hand, we now contrast $\mathsf{PostQMA}^{*}$ with another generalisation of $\mathsf{QMA}$.

Lin and Fefferman described the class $\mathsf{QMA_{exp}}$ in Ref. \cite{fefferman2016quantum}. This class $\mathsf{QMA_{exp}}$ is defined in the same way as $\mathsf{QMA}$, except the gap between $\alpha$ and $\beta$ now satisfies $\alpha-\beta\geq 2^{-poly(n)}$ for $poly$ being some polynomial in the input size $n$. It was proven by Lin and Fefferman that $\mathsf{QMA_{exp}}=\mathsf{PSPACE}$, the class of problems decided by a deterministic classical computer using an amount of space at most polynomial in the size of the input. In addition they also described another natural generalisation of the \textsc{k-Local Hamiltonian} problem, where the promise gap $b-a$ is permitted to be separated by only an inverse exponential, i.e. $b-a\geq 2^{-poly(n)}$ for some polynomial function $poly$. This problem is called the \textsc{Precise k-Local Hamiltonian} problem and was shown to be complete for $\mathsf{QMA_{exp}}$.

How does $\mathsf{QMA_{exp}}$ relate to $\mathsf{PostQMA}^{*}$? It is known that $\mathsf{PP}\subseteq\mathsf{PSPACE}$, with good evidence that equality does not hold \citep{aspnes}. Therefore, with this computational complexity evidence at hand we could argue that the \textsc{k-Local Hamiltonian} problem constructed from tame post-selected quantum circuits will not be as hard to solve as an arbitrary \textsc{Precise k-Local Hamiltonian} problem. It seems that the Hamiltonians resulting from tame post-selected quantum circuits could live in an intermediate regime between $\mathsf{QMA}$ and $\mathsf{QMA_{exp}}$.

In this section we have given some evidence from computational complexity that the Kitaev-Feynman circuit-to-Hamiltonian circuit when applied to tame post-selected quantum circuits can result in encoding hard computations in the ground-space. However, this computation is still bounded in a sensible way, thus reinforcing the notion of it being tame. In particular, tame post-selected quantum circuits will probably not have the power of $\mathsf{QMA_{exp}}$ even though we are permitted to post-select on events which occur with an exponentially small probability. In the next section we will look at families of tame post-selected circuits and their corresponding Hamiltonians, and in particular numerically study the gap between the ground-state and the next highest energy state.

\section{PostIQP and Hadamard Gadgets}\label{sec4}

One surprising aspect of post-selection is that it can limit the kinds of circuits we need to consider to get the complexity class $\mathsf{PostBQP}$. In particular, the circuit can be an \textit{Instantaneous Quantum Polytime} (IQP) circuit \citep{bremner2010classical}, which consists of preparing a set of qubits in the state $|0\rangle$, applying a set of unitary gates that are diagonal in the local Pauli-$X$ basis, and then measuring in the computational basis. Since all of the unitary gates are diagonal in the same basis, they all commute with each other and can be considered to be implemented ``simultaneously'', in some sense. Equivalently, an IQP circuit consists of preparing a set of qubits in the state $|+\rangle$, applying a set of unitary gates that are diagonal in the local Pauli-$Z$ basis, then Hadamard gates to all of the qubits, and finally measuring qubits in the computational basis. To be a uniform IQP circuit, given a classical input $x\in\{0,1\}^{n}$, the description of the gates which are diagonal in the Pauli-$Z$ basis must be generated by a classical computer in time at most polynomial in $n$. By considering postselection, the complexity class, $\mathsf{PostIQP}$ is obtained, for which Bremner, Jozsa, and Shepherd proved that $\mathsf{PostIQP}=\mathsf{PostBQP}$, the proof of which is based on the Hadamard gadget as outlined earlier. Thus, post-selection drastically simplifies the kinds of circuits we need to consider. 

The result of all of this is that we can define $\mathsf{PostQMA}$ such that Arthur's circuit is an IQP circuit, and by virtue of $\mathsf{PostQMA}^{*}=\mathsf{PostBQP}$ we have exactly the same computational power. Therefore we can restrict to considering Hamiltonians constructed from IQP circuits without loss of generality. Taking these circuits as inspiration we will consider how the gap $b-a$ of the Hamiltonian scales in the size of the post-selected circuit with which we start.

We will consider two classes of post-selected circuit based on the Hadamard gadget, and numerically analyse the corresponding gaps $b-a$ between the ground state energy and the energy of the first excited state. The Hadamard gadget satisfies the notion of tame post-selection so is an ideal candidate for which we can build circuits. It should be noted that in both families of circuit the total probability of success of the post-selected event decreases exponentially in the size of the circuit. However, in terms of the gap $b-a$, in one family denoted as $\mathcal{F}_{1}$ this appears to decrease exponentially in the size of the post-selected circuit, but in the other family, denoted by $\mathcal{F}_{2}$, it seems to decrease polynomially in the size of the circuit. Thus, the intuition that the probability of success dictates the gap $b-a$ of the corresponding Hamiltonian is not immediately obvious. To emphasize this point, in both families of circuits the corresponding Hamiltonians all have terms that have operator norms that are bounded by some polynomial in the circuit size. Therefore, the circuit family $\mathcal{F}_{2}$ seems to encode a Hamiltonian problem contained in $\mathsf{QMA}$, while family $\mathcal{F}_{1}$ does not.

\subsection{Circuit family $\mathcal{F}_{1}$}

First, we consider an arbitrary quantum state and $n$ post-selection qubits initialised in the $|+\rangle$ state.  Neighbouring qubits are entangled with a CZ gate and are then measured one after the other in the Pauli-$X$ basis and post-selecting upon receiving outcome $|+\rangle$. An example of such as circuit for three gates is shown in Fig. \ref{Cascade}. Effectively, the state of the first qubit is teleported on the second and acted upon by a Hadamard gate. This is then the input to a new Hadamard gadget, implementing a new Hadamard gate. The effect of this circuit is to sequentially teleport the state $|\psi \rangle$ from qubit to qubit, each time applying a Hadamard gate to it, thus causing it to oscillate between the state $|\psi \rangle$ and $H|\psi \rangle$.
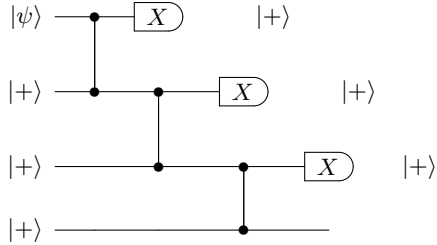
\begin{figure}[h!]
	\centering
	\mbox{
		\Qcircuit @C=1.5em @R=1.9em {
			\lstick{|\psi \rangle }	& \ctrl{1} 	   &   \measureD{X} & \rstick{\ket{+}}	\\
			\lstick{|+\rangle }	& \ctrl{-1} &\ctrl{1} 	  &  \measureD{X} & \rstick{\ket{+}}\\
			\lstick{|+\rangle }	& \qw 		 & \ctrl{-1} & \ctrl{1}  &  \measureD{X} & \rstick{\ket{+}}\\
			\lstick{|+\rangle }	& \qw  		 & \qw 		& \ctrl{-1} & \qw  & 
		} }
		\caption{Three Hadamard gadgets are implemented using three additional qubits.} 
		\label{Cascade}
	\end{figure}

If $n$ Hadamard gates are applied, then we need to implement $n$ Hadamard gadgets, which requires$n$ ancillary post-selection qubits and $n$ measurements. The space of our qubits will be $2^{n+1}$, and the clock will be qudit of dimension $2n+1$. The propagation Hamiltonian will be made of $2n$ terms, where odd terms correspond to Kitaev's unitary Hamiltonians and even terms to a projection Hamiltonian. Explicitly, we can write the propagation Hamiltonian as: 
\begin{align}\label{mayhamexp}
\begin{split}
H_{prop}=&\sum_{j=0}^{n-1} \frac{1}{2}\Big[-CZ^{(j+1, j+2)} \otimes \big(  |2j\rangle \langle 2j+1|+ |2j+1\rangle \langle 2j|  \big)  
+ \mathbb{I} \otimes  ( |2j\rangle \langle 2j|+ |2j+1\rangle \langle 2j+1|) \Big] \\
&+\frac{1}{3}\Pi^{(j+1)} \otimes \Big[ 2 |2j+1\rangle \langle 2j+1| - \frac{1}{\sqrt{2}} |2j+1\rangle \langle 2j+2|
-\frac{1}{\sqrt{2}}|2j+2\rangle \langle 2j+1| + |2j+2\rangle \langle 2j+2| \Big]\\
&+\left(\mathbb{I}-\Pi\right)^{(j+1)} \otimes |2j+2\rangle \langle 2j+2|,
\end{split}
\end{align}
where $CZ^{(i, j)}$ denotes the control-$Z$ operator acting on qubits $i$ and $j$ with identity on all others, and where $\Pi^{(i)}$ corresponds to the projector $|+\rangle \langle +|$ acting on qubit $i$, with identity on all other qubits. 

 We constructed the propagation Hamiltonian of the circuit, and computed its smallest non-zero eigenvalue, and illustrate its reciprocal $\lambda^{-1}_{min}$ in Fig. \ref{Exponential Scaling}, with an exponential function fitted to the data. Therefore, the intuition that as the probability of success decreases exponentially, the gap closes as an inverse exponential seems to be correct.
\begin{figure}[h!]
\centering 
\includegraphics[scale=0.45]{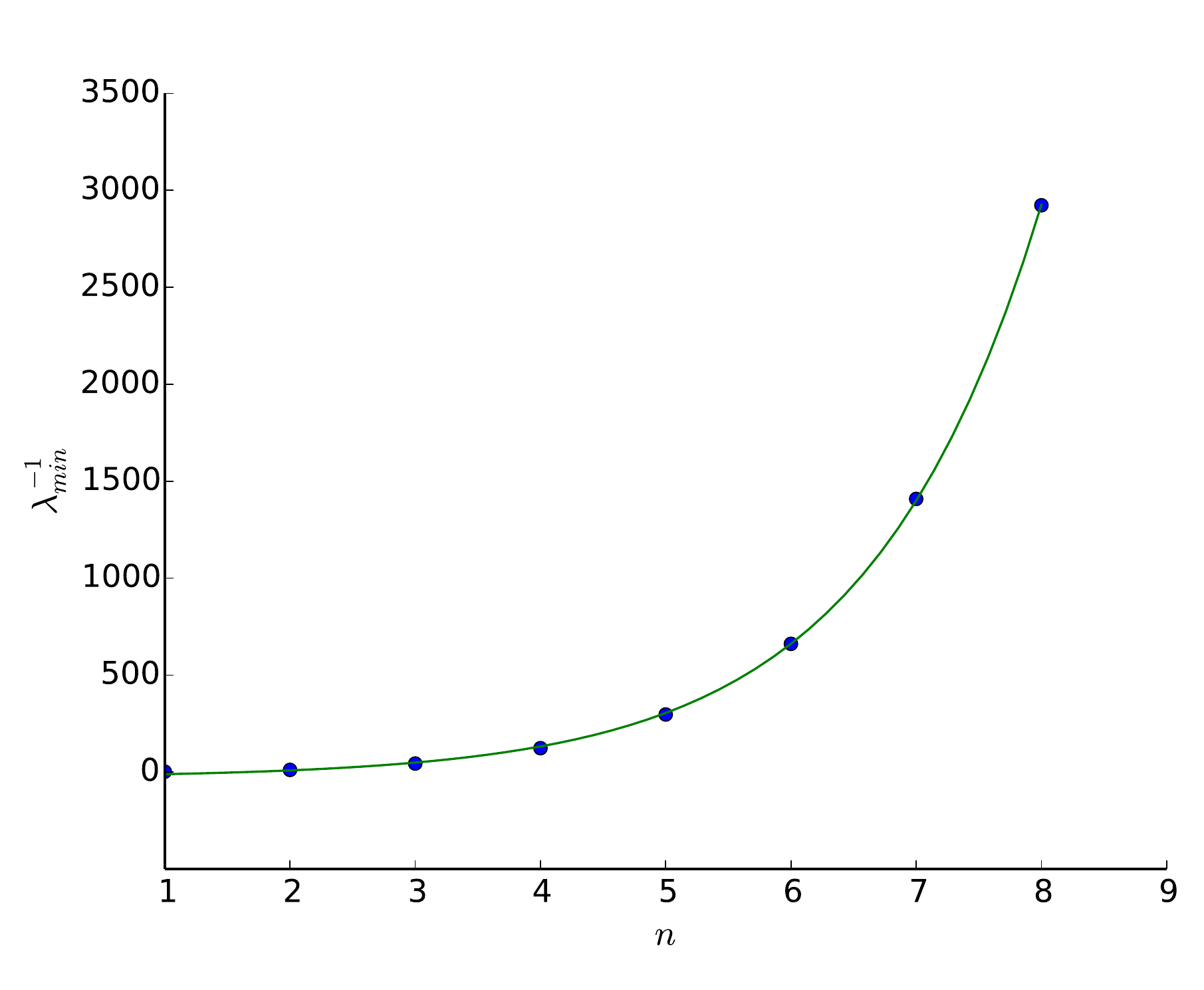}
\caption[Exponential Scaling.]{Scaling of the eigenvalues of the Hamiltonian as a function of Hadamard gadgets, with Hamiltonian in equation \ref{mayhamexp}. Here, we fit the data to an exponential function $y=Ae^{bc} +c$, yielding $A=8.802$, $b=0.727$ and $c=-28.767$.}
\label{Exponential Scaling}
\end{figure}

\subsection{Circuit family $\mathcal{F}_{2}$}

In this circuit family, again the input to the circuit is $|\psi \rangle$ with an ancillary qubit initialised as $|+\rangle$. One round of the circuit will correspond to the application of the controlled-$Z$ gate, a measurement of the first qubit in the Pauli-$X$ basis with post-selection on outcome $|+\rangle$, followed by another controlled-$Z$ gate and a final measurement on the second qubit in the Pauli-$X$ basis with post-selection on outcome $|+\rangle$.  With post-selection, this initial circuit effectively implements the identity, as the output is given by $|\psi \rangle \otimes |+\rangle$. By repeating this process, we obtain a circuit such as that in Fig. \ref{Trivial.}, where each box corresponds to this initial circuit involving two post-selections.
\begin{figure}[h!]
	\centering
	\mbox{
		\Qcircuit @C=1.3em @R=2.2em {
			\lstick{|\psi \rangle }	&  \qw             & \ctrl{1}  \ar@{.}[]+<0em,1em>;[d]+<0em,-2em> &   \measureD{|+\rangle} \ar@{.}[]+<0em,1em>;[d]+<0em,-2em> & \ctrl{1} \ar@{.}[]+<0em,1em>;[d]+<0em,-2em>    & \qw 	\ar@{.}[]+<0em,1em>;[d]+<0em,-2em>             & \ctrl{1}  \ar@{.}[]+<0em,1em>;[d]+<0em,-2em>  &  \measureD{|+\rangle} \ar@{.}[]+<0em,1em>;[d]+<0em,-2em> & \ctrl{1}\ar@{.}[]+<0em,1em>;[d]+<0em,-2em> &  \qw \ar@{.}[]+<0em,1em>;[d]+<0em,-2em> \ar@{.}[]+<0em,1em>;[d]+<0em,-2em>& \ar@{.}[];\\
			\lstick{|+\rangle }	     & \qw  & \ctrl{-1} &   \qw  	            &  \ctrl{-1} & \measureD{|+\rangle}   & \ctrl{-1}  & \qw                   & \ctrl{-1} & \measureD{|+\rangle} & \text{ }  \gategroup{1}{2}{2}{6}{.7em}{--}  \gategroup{1}{7}{2}{10}{.7em}{--} \\
			&  & 1 & 2 & 3 & 4 & 5 & 6 & 7 & 8 
		} }
		\caption{Each box is a post-selected circuit implementing the identity on the input qubit $|\psi\rangle$ and recycling qubits since the ancillas will always be in the state $|+\rangle$.} 
		\label{Trivial.}
	\end{figure}
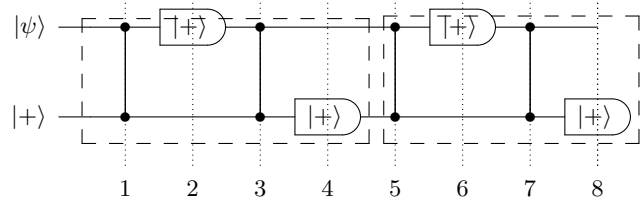			

The propagation Hamiltonian $H_{prop}$ in the Kitaev-Feynman construction corresponding to this circuit is now constructed according to the number of rounds $n$ of each gadget in the box of Fig. \ref{Trivial.}. One box corresponds to a unitary operator being applied, a renormalised projector, a unitary operator, and a final renormalised projector $H_{prop}=\sum_{j=0}^{n-1} H_j$, where

\begin{align}\label{mayham}
\begin{split}
H_j=&\frac{1}{2}(-CZ \otimes (|j\rangle \langle j+1|  +|j+1\rangle \langle j|)   
 + \mathbb{I} \otimes  (|j\rangle \langle j| +  |j+1\rangle \langle j+1|) 
+\frac{1}{3}(|+\rangle \langle +|\otimes \mathbb{I}) \otimes \Big[ 2 |j+1\rangle \langle j+1| \\
&- \frac{1}{\sqrt{2}} |j+1\rangle \langle j+2|-\frac{1}{\sqrt{2}}|j+2\rangle \langle j+1| + |j+2\rangle \langle j+2| \Big] +(|-\rangle \langle -|\otimes \mathbb{I}) \otimes |j+2\rangle \langle j+2|\\
&+\frac{1}{2}(-CZ \otimes (|j+2\rangle \langle j+3| + |j+3\rangle \langle j+2|)   
 + \mathbb{I} \otimes  (|j+2\rangle \langle j+2| +  |j+3\rangle \langle j+3|) \\
&+\frac{1}{3}(\mathbb{I}\otimes |+\rangle \langle +|) \otimes \Big[ 2 |j+3\rangle \langle j+3| - \frac{1}{\sqrt{2}} |j+3\rangle \langle j+4|
-\frac{1}{\sqrt{2}}|j+4\rangle \langle j+3| + |j+4\rangle \langle j+4| \Big]\\
&+(\mathbb{I}\otimes |-\rangle \langle -|) \otimes |j+4\rangle \langle j+4|.
\end{split}
\end{align}
The dimension of the auxiliary clock depends on the number of rounds we implement and is given by $4n+1$. We numerically find the smallest non-zero eigenvalue and depict its reciprocal $\lambda^{-1}_{min}$ as it scales in $n$ in Fig. \ref{Polynomial Scaling.}, with a quadratic function fitted to the data. The smallest non-zero eigenvalue of the propagation Hamiltonian thus seems to scale as an inverse polynomial function.

\begin{figure}[h!]
\centering 
\includegraphics[scale=0.45]{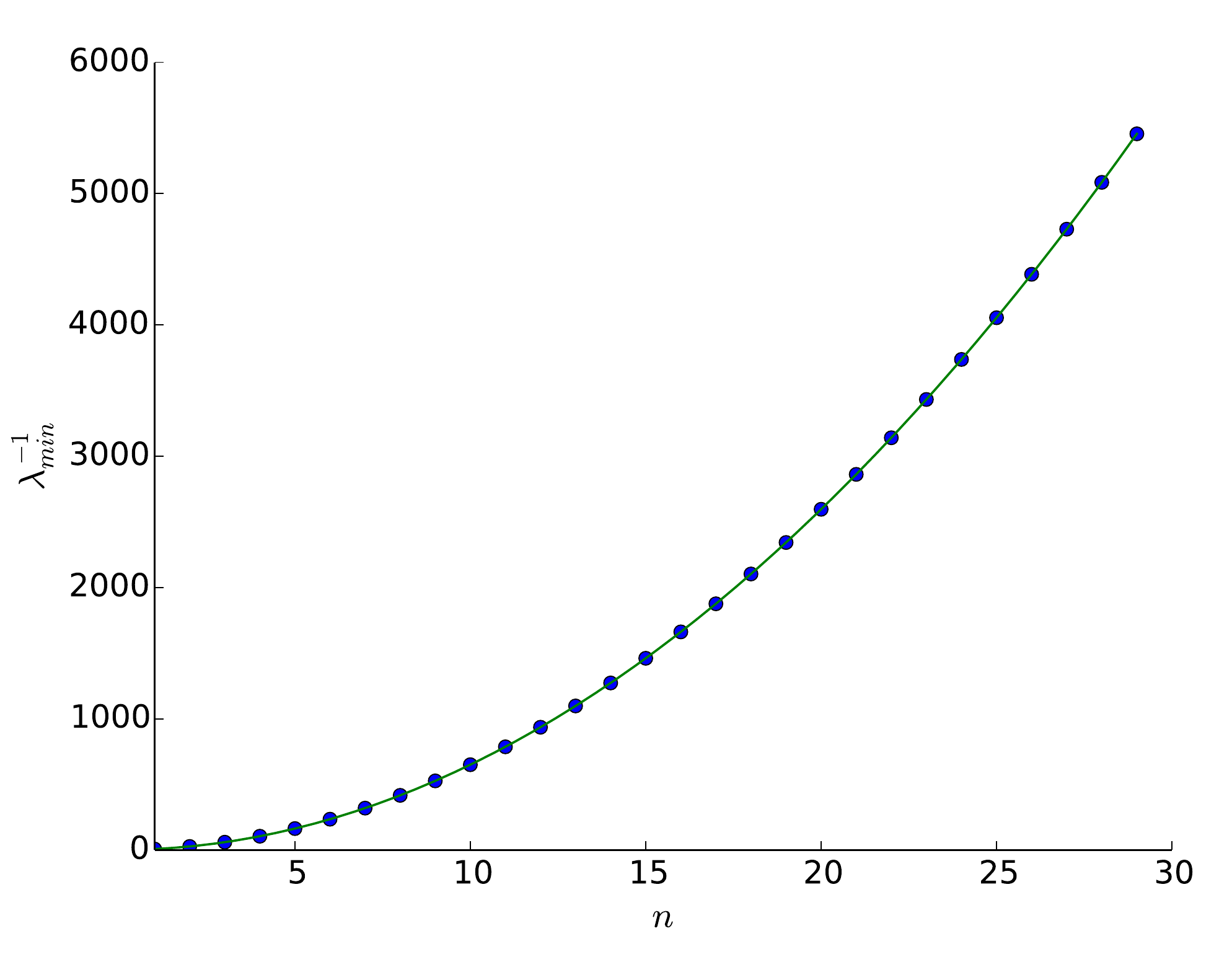}
\caption[Polynomial Scaling]{Scaling of the reciprocal of the smallest eigenvalue of the Hamiltonian corresponding to circuit family $\mathcal{F}_{2}$, as shown in equation \ref{mayham}. Here, we fit the data to a quadratic function $y=ax^2 +bx +c$, and obtain $a=6.5$, $b=0.04$ and $c=1.4$.}
\label{Polynomial Scaling.}
\end{figure}

Thus, we have two quantum circuits effectively implementing the same trivial operation, but which nonetheless exhibit a starkly different behaviour. In each case, the Hamiltonian encoding the circuit is built, and its smallest non-zero eigenvalue plotted as a function of system size. Two starkly distinct cases are observed. In the first, the eigenvalues scales exponentially with system size, whereas in the second, a polynomial --- quadratic --- scaling is observed, even though both circuits are effectively implementing the same unitary operator. Yet, the spectrum of their respective Hamiltonians are radically different. 
\section{Discussion}\label{sec5}
In this work, we considered the \textsc{k-Local Hamiltonian} problem, and in particular the circuit-to-Hamiltonian construction used in the proof of $\mathsf{QMA}$-completeness due to Kitaev. We introduced a new class of Hamiltonians encoding non-unitary computation and postselection, and investigated the scaling of its smallest non-zero eigenvalue with system size. This was achieved by extending the circuit-to-Hamiltonian construction to evolutions via renormalised projectors, which map pure states to pure states. In order for these Hamiltonians to not depend on the input state, we introduce the idea of a restricted form of postselection, which we call tame postselection, where the probability of an event occurring is input independent. We then considered the computational complexity of the computations that are being encoded in this extended construction, and showed that they are exactly the quantum computations with post-selection as defined by Aaronson. Therefore, given certain assumptions about computational complexity, solving the \textsc{k-Local Hamiltonian} problem given Hamiltonians constructed from circuits with post-selection is harder than the standard problem that is $\mathsf{QMA}$-complete, but not as hard as the so-called \textsc{Precise k-Local Hamiltonian} problem which allows for gap that is exponentially small between energy eigenstates.

The main direction for future research is to get a better characterisation of Hamiltonians resulting from post-selected quantum circuits. We numerically explored a couple of examples of post-selected circuit families that exhibited similar behaviour from the point-of-view of state transformation and the probability of success of all post-selected events exponentially decreased in the size of the circuits, however their corresponding Hamiltonians exhibited very different behaviour. It seems that one of the Hamiltonians can be solved within $\mathsf{QMA}$ since there was a polynomially small gap between the ground state and first excited state energies, the other family of Hamiltonians seemed to have an exponentially small gap. Therefore, this gap might not be determined by the probability of success for the post-selected events nor the effective unitary implemented by the tame post-selection as indicated by Prop. \ref{prop1}. The natural question is then what determines this gap?

One major hope for this work is that it is useful in demonstrating that the simulation of certain Hamiltonians is hard, such as in the work of Ref. \cite{bouland}. Since post-selection is a useful tool in proving such hardness results, it seems natural to build post-selection into the Hamiltonians and then make arguments based on the \textsc{k-Local Hamiltonian} problem. By bringing all of these elements together we may get a better understanding of what kinds of quantum systems are hard to classically simulate and why.

\textit{Acknowledgements} - This work was supported by EPSRC, partly through the Centre for Doctoral Training in Delivering Quantum Technologies [EP/L015242/1], and the Networked Quantum Information Technologies (NQIT) Hub [EP/M013243/1].

 \onecolumngrid  
\end{document}